\documentclass[%
 reprint,
 superscriptaddress,
 showpacs,
 showkeys,
 amsmath,
 amssymb,
 aps,
 prx,
 floatfix,
]{revtex4-1}

\usepackage{amsmath}                    
\usepackage{amsthm}                     
\usepackage{graphicx}                   
\usepackage[makeroom]{cancel}			

\newtheorem{theorem}{Theorem}
\newtheorem{corollary}{Corollary}[theorem]

\newcommand{\pd}{\partial}                  
\newcommand{\dd}{\mathrm{d}}                
\newcommand{\ds}{\displaystyle}             

\newcommand{\M}{\mathrm{M}}                 
\newcommand{\K}{\mathcal{K}}                
\newcommand{\h}{\mathrm{H}}                 
\newcommand{\en}{\mathrm{E}}                

\newcommand{\ads}{\mathrm{arcds}}
\newcommand{\const}{\mathrm{const}}

\begin{document}

\title{Betatron frequency and the Poincar\'e rotation number}
\author{Sergei Nagaitsev}
\affiliation{Fermilab, Batavia, IL 60510, USA}
\affiliation{The Enrico Fermi Institute, The University of Chicago, Chicago, IL
60637, USA}
\author{Timofey Zolkin}
\affiliation{Fermilab, Batavia, IL 60510, USA}

\date{\today}

\begin{abstract}
Symplectic maps are routinely used to describe single-particle dynamics in
circular accelerators.
In the case of a linear accelerator map, the rotation number (the betatron
frequency) can be easily calculated from the map itself.
In the case of a nonlinear map, the rotation number is normally obtained
numerically, by iterating the map for given initial conditions, or through a
normal form analysis, a type of a perturbation theory for maps.
Integrable maps, a subclass of symplectic maps, allow for an analytic evaluation
of their rotation numbers.
In this paper we propose an analytic expression to determine the rotation number
for integrable symplectic maps of the plane and present several examples, relevant to accelerators.
These new results can be used to analyze the topology of the accelerator
Hamiltonians as well as to serve as the starting point for a perturbation theory
for maps.
\end{abstract}

\maketitle

\section{Introduction}

The first mention of the betatron frequency was in the 1941 pioneering work by
Kerst and Serber~\cite{Kerst_PhysRev.60.53}, where they defined it as the
fractional number of particle oscillations around the orbit per one revolution
period in a betatron (a type of an induction accelerator).
Later, the theory of the alternating-gradient (AG)
synchrotron~\cite{courant1958theory} demonstrated the existence of an integral
of motion (the so-called Courant-Snyder invariant) for particles in an AG
synchrotron and established a powerful connection between the modern AG focusing
systems and linear symplectic maps, thus connecting the betatron frequency and
the Poincar\'e rotation number~\cite{Poincare}. 

In modern accelerators (for example, in the LHC) particles are stored for up to 
$10^9$ revolutions and understanding their dynamics is crucially important
for maintaining long-term particle stability~\cite{Todesco1996, Papaphil}.
One important parameter of particle dynamics in an accelerator is the betatron
frequency and its dependence on a particle's amplitude.
It turns out that the accelerator focusing systems conserve the Courant-Snyder
invariant only approximately and there is a need to analyze the conditions for
stable particle dynamics.
Over the recent years, several methods were developed to analyze the particle
motion in accelerator systems, using either numeric tools, like the Frequency
Map Analysis~\cite{Laskar}, or the
Normal Form Analysis~\cite{Bazzani:1994ks, Turchetti}, a type of a perturbation
theory, which uses a linear map and a Courant-Snyder invariant as a starting
point.

At the same time, there has been continuous interest, starting with
E.~McMillan~\cite{mcmillan1971problem}, in making the accelerator maps nonlinear,
yet integrable~\cite{Cary,Danilov1,DN1, DanNagNBody}.
However, there does not exist an analytic method to calculate the betatron
frequency (the Poincar\'e rotation number) for nonlinear symplectic integrable
maps. This present paper is set to remedy this deficiency.

\section{Betatron frequency}

For a one degree-of-freedom time-independent system, the Hamiltonian function,
$\h[p,q;t] = \en$, is the integral of the motion.
If the motion is bounded, it is also periodic, and the period of oscillations
can be determined by integrating
\begin{equation}
\label{math:T}
    T(\en) = \oint\left(\frac{\pd \h}{\pd p}\right)^{-1}\,\dd q,
\end{equation}
where $p=p(\en,q)$~\cite{Lichtenberg1992}.
The oscillation period and its dependence on initial conditions is one of the
key properties of the periodic motion.

Let us now consider a symplectic map of the plane
(corresponding to a one-turn map of an accelerator), $\M:\mathbb{R}^2\rightarrow\mathbb{R}^2$,
\[
    (q',p') = \M\,(q,p),
\]
where the prime symbols ($'$) indicate the transformed phase space coordinates.
Suppose that the sequence, generated by a repeated application of the map,
\[
    (q_0,p_0)\rightarrow(q_1,p_1)\rightarrow
    (q_2,p_2)\rightarrow(q_3,p_3)\rightarrow
    \ldots
\]
belongs to a closed invariant curve.
We do not describe how this map is obtained (see, for example,~\cite{Intro}) but
let us suppose that we know the mapping equations.
Let $R_n$ be the rotation angle in the phase space $(q, p)$ around a stable
fixed point between two consecutive iterations $(q_n,p_n)$ and 
$(q_{n+1},p_{n+1})$.
Then, the limit, when it exists, 
\begin{equation}
\label{math:nu}
	\nu = \lim_{N \to \infty} \frac{1}{2\,\pi\,N}\,\sum_{n=0}^{N} R_n 
\end{equation}
is called the rotation number (the betatron frequency of the one-turn map) for
that particular orbit of the map $\M$~\cite{Dilao1996}.
Unlike Eq.~(\ref{math:T}), which allows to express the oscillation period
analytically, Eq.~(\ref{math:nu}) can be only evaluated numerically for each
orbit. 
Let us now suppose that there exists a non-constant real-valued continuous
function $\K(q,p)$, which is invariant under $\M$.
The function $\K(q,p)$ is called {\it integral} and the map is called
{\it integrable}.
In this paper, we are describing the case, for which the level sets $\K=\const$
are compact closed curves (or sets of points) and for which the identity
\begin{equation}
\label{math:int_map}
    \K(q',p') = \K(q,p)
\end{equation}
holds for all $(q,p)$.
There are many examples of integrable maps, including the famous McMillan
map~\cite{mcmillan1971problem}, described below.
The dynamics is in many  ways similar to that of a continuous system, however,
Eq.~(\ref{math:T}) is not directly applicable since the integral $\K(q,p)$ is
not the Hamiltonian function.
Below, we will present an expression (the Danilov theorem) to obtain the
rotation number from $\K(q,p)$ for an integrable map, $\M$.

The {\bf Arnold-Liouville theorem} for integrable
maps \cite{arnold1968ergodic,veselov1991integrable,meiss1992symplectic}
states that (1) the action-angle variables exist and (2) in these variables,
consecutive iterations of integrable map $\M$ lie on nested circles of radius
$J$ and that the map can be written in the form of a twist map,
\begin{equation}
\label{math:twist}
    \begin{bmatrix}
        J_{n+1}	\\[0.2cm]  \theta_{n+1}
    \end{bmatrix}
    =
    \begin{bmatrix}
        J_n     \\[0.2cm] \theta_n + 2\,\pi\,\nu(J) \mod 2\,\pi
    \end{bmatrix},
\end{equation}
where $|\nu(J)| \leq 0.5$ is the rotation number, $\theta$ is the angle variable
and $J$ is the action variable, defined by the map $\M$ as
\begin{equation}
\label{math:J}
	J = \frac{1}{2\,\pi}\,\oint p\,\dd q.
\end{equation}
Thus, in this paper, we would like to consider the following question: how does
one determine the rotation number, $\nu(\K)$, from the known integral, 
$\K(q,p)$, and the known integrable map, $\M$?
In addition, in the "Examples" section we propose how to use this theorem
when only an approximate invariant is known.

\section{Danilov theorem}

\begin{theorem}[{\bf Danilov theorem}~\cite{danilovPRIVATE}] 
Suppose a symplectic map of the plane,
\[
    (q',p') = \M\,(q,p),
\]
is integrable with the invariant (integral) $\K(q, p)$, then its Poincar\'e rotation
number is
\begin{equation}
\label{math:Danilov}
\nu(\K) =
    \int_q^{q'}
	    \left(\frac{\pd \K}{\pd p}\right)^{-1}\,\dd q
	\Bigg/
	\oint
	    \left(\frac{\pd \K}{\pd p}\right)^{-1}\,\dd q,
\end{equation}
where the integrals are evaluated along the invariant curve, $\K(q, p)$.
\end{theorem}

\begin{figure}[th!]\centering
\includegraphics[width=\linewidth]{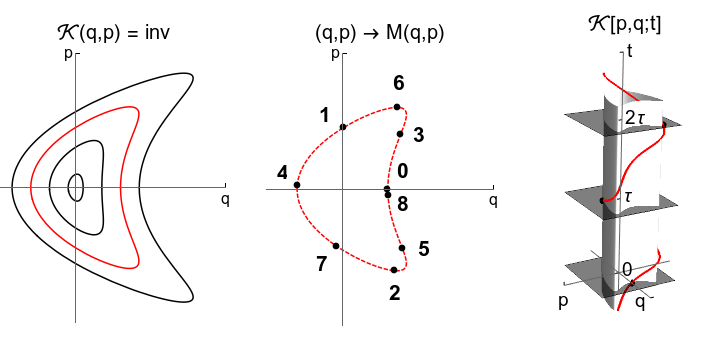}
\caption{\label{fig:DanilovTh}
	Constant level sets of the integral $\K(q,p)=\const$ (left).
	A particular curve representing a level set of $\K$
	and several iterates of the map $\M$ (center).
	A three-dimensional phase space, $(q,p)$ + time, of the 
	system~(\ref{math:Kham}) (right).
	Dark gray planes $t=0,\tau,2\tau,\ldots$ represent stroboscopic 
	Poincar\'e section of the continuous flow of the system 
	(red curve) which is identical to map $\M$.
	}
\end{figure}

\begin{proof} 
Consider the following system of differential equations:
\begin{equation}
\label{math:Kham}
	\frac{\dd Q}{\dd t} =  \frac{\pd\K(Q,P)}{\pd P}, \qquad\qquad
	\frac{\dd P}{\dd t} = -\frac{\pd\K(Q,P)}{\pd Q}.
\end{equation}
We notice that $\K(Q,P)$ does not change along a solution of the system, because it is an integral of the motion, meaning
\begin{equation}
\label{math:Intg}
	\frac{\dd \K}{\dd t} = \frac{\pd\K}{\pd Q} \frac{\dd Q}{\dd t} + \frac{\pd\K}{\pd P} 	\frac{\dd P}{\dd t} = 0
\end{equation}
for any solution $Q(t)$ and $P(t)$.
Let $q(t)$ and $p(t)$ be the solutions of the system~(\ref{math:Kham}) with the
following initial conditions $q(0)=q_0$ and $p(0)=p_0$.
Define a new map, $\widetilde{\M}(q,p)$ (see Fig.~\ref{fig:DanilovTh})
\begin{equation}
    (q',p') =
    \widetilde{\M} (q,p) =
    \left(q(\tau),p(\tau)\right) 
\end{equation}
where $\tau$ is a discrete time step.
For a given $\K$, which is an integral of both $\M$ and $\widetilde{\M}$, one
can always select $\tau(\K)$ such that the maps $\M(q,p)$ and
$\widetilde{\M}(q,p)$ are identical.
This follows from the Arnold-Liouville theorem.
Since $\K(q,p)$ is compact and closed, the functions $q(t)$ and $p(t)$ are 
periodic with a period $T(\K)$.
By its definition,
\begin{equation}
    \tau = \nu(\K)\,T(\K).
\end{equation}
Let us now calculate $\nu(\K)$:
\begin{equation}
\label{math:Dan6}
\begin{array}{l}
\displaystyle
    \nu(\K) \equiv \frac{\tau}{T} = 
    \frac{\int_{q}^{q'}\,\dd t}{\oint\,\dd t} =
    \frac{\int_{q}^{q'}
        \left( \frac{\dd q}{\dd t} \right)^{-1}
    \,\dd q}
     {\oint
        \left( \frac{\dd q}{\dd t} \right)^{-1}
    \,\dd q}                        \\[0.65cm]
\displaystyle \qquad\,\,
    = \frac{\int_{q}^{q'}
        \left(\frac{\pd \K}{\pd p}\right)^{-1}
    \,\dd q}
    {\oint
        \left(\frac{\pd \K}{\pd p}\right)^{-1}
    \,\dd q} .
\end{array}
\end{equation}
Q.E.D.
\end{proof}

\begin{figure}[t]\centering
\includegraphics[width=\linewidth]{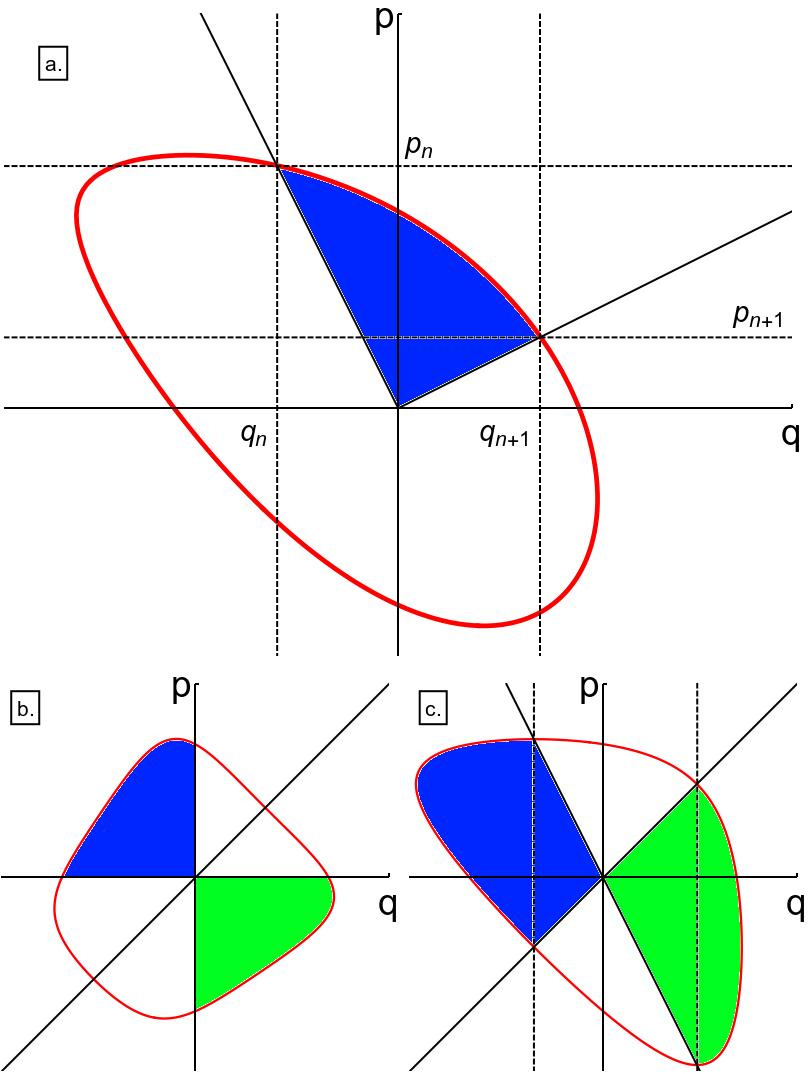}
\caption{\label{fig:partial_action}
	The partial action is defined as a sector area (blue) for one map iteration, divided by ${2\,\pi}$ (a.).
	Convenient choices of the partial action for mappings in McMillan form: an area under the curve in II (blue) or IV (green) quadrants (b.), and, areas for initial conditions in a form of $(q_0,q_0)$ (c.).
	} 
\end{figure}

\begin{corollary}

\begin{equation}
\label{math:Dan2}
    \nu(\K)=\frac{\dd J'}{\dd J},
\end{equation}
where 
\begin{equation}
    J'(\K) = \frac{1}{2\,\pi}\,\int_q^{q'} p(\K,q)\,\dd q.
\end{equation}
is the partial action calculated as a sector integral
(see Fig.~\ref{fig:partial_action}) around the stable fixed point.
\end{corollary}

\begin{proof}
First, we will consider the denominator in Eq. (\ref{math:Dan6}):
\begin{equation}
\label{math:Dan4}
   \frac{1}{2\,\pi}\,\oint \left(\frac{\pd \K}{\pd p}\right)^{-1}\,\dd q = \frac{1}{2\,\pi}\,\frac{\dd}{\dd \K}\,\oint p\,\dd q = \frac{\dd J}{\dd \K}.
\end{equation}
Second, we will evaluate the numerator.  Using the equations of motion in
Eq.~(\ref{math:Kham}), we notice that 
\begin{equation}
\label{math:Dan3}
   \int_q^{q'} \left(\frac{\pd \K}{\pd p}\right)^{-1}\,\dd q = - \int_p^{p'} \left(\frac{\pd \K}{\pd q}\right)^{-1}\,\dd p .
\end{equation}
Now, we will utilize the Leibniz integral rule together with
Eq.~(\ref{math:Dan3}) to obtain
\begin{equation}
\label{math:Dan5}
\begin{array}{l}
\displaystyle
     \frac{1}{2\,\pi}\,\int_q^{q'}
        \left(\frac{\pd \K}{\pd p}\right)^{-1}
    \,\dd q =
     \frac{1}{2\,\pi}\times                     \\[0.65cm]
\displaystyle
    \,\,\,\,\,\,\frac{\dd }{\dd \K}
        \left(  \frac{q\,p - q'\,p'}{2} +
                \int_q^{q'} p\,\dd q
        \right) = \frac{\dd J'}{\dd \K}.
\end{array}
\end{equation}
Finally, by combining Eqs.~(\ref{math:Dan4}) and (\ref{math:Dan5}) we obtain the
Eq.~(\ref{math:Dan2}).
\end{proof}

\begin{corollary}
For a linear map ($\nu = \const$),
\begin{equation}
\label{math:Dan7}
    \nu = J'/J.
\end{equation}
\end{corollary}
\begin{proof}
Since $\nu = \const$, the Hamiltonian function is $\h(J)= \nu\,J$.
Using Eq.~(\ref{math:Dan2}), we obtain Eq.~(\ref{math:Dan7}).
\end{proof}

\begin{corollary}
The Hamiltonian function corresponding to the map $M$ is 
\begin{equation}
\label{math:hamilt}
    \h(\K) = J'(\K).
\end{equation}
\end{corollary}
\begin{proof}
Since $\nu = \dd\h/\dd J$, one can use Eq.~(\ref{math:Dan2}) to obtain
$\h = J' + \const$.
\end{proof}

\begin{corollary}
\begin{equation}
\label{math:Danilov1}
    \nu(\K) = \int_{p}^{p'}
	    \left(\frac{\pd \K}{\pd q}\right)^{-1}
	\,\dd p
	\Bigg/
	\oint
	    \left(\frac{\pd \K}{\pd q}\right)^{-1}
	\,\dd p,
\end{equation}
where the integrals are evaluated along the invariant curve, $\K(q, p)$.
\end{corollary}
\begin{proof}
Because of the $p \leftrightarrow -q$ symmetry in Eqs.~(\ref{math:Kham}), the proof
is similar to Eq. (\ref{math:Dan6}).
\end{proof}

In order to generalize the Danilov theorem to higher-dimensional integrable maps,
one has to know the variables, where such a map is separated into maps for each degree of freedom.
Below we will consider an example of a 4D map, which is separable in polar coordinates
with two integrals of motion.


\section{Examples}

In order to employ this theorem in practice, one would need to recall that with
$p=p(\K,q)$, the integrand in Eq.~(\ref{math:Danilov}) is
\begin{equation}
	\left( \frac{\pd\K}{\pd p} \right)^{-1}=\frac{\pd p(\K,q)}{\pd \K}.
\end{equation}
Also, the lower limit of the integral can be chosen to be any convenient value
of $q$, for example 0, as long it belongs to a given level set, $\K(q,p)$.
Finally, the upper limit of the integral, $q'$, is obtained from the
selected $q$ and $p=p(\K,q)$ by iterating the map, $\M(q,p)$.  It is clear that
not all functions $\K(q,p)$ can be inverted analytically to obtain $p=p(\K,q)$.  
This drawback of this method can be overcome by numeric evaluations (see Appendix B). 

For maps in a special (McMillan) form \cite{mcmillan1971problem},
\begin{equation}
\label{math:McMap1}
    \begin{bmatrix}
        q'	\\ p'
    \end{bmatrix}
    =
    \begin{bmatrix}
        p     \\ -q + f(p)
    \end{bmatrix},
\end{equation}
the convenient choices for integration limits in Eq.~\ref{math:Danilov} are
$(q, p) = (q_0, 0)$ and $(q', p') = (0, -q_0+f(0))$,
Fig.~\ref{fig:partial_action}.b, and 
$(q, p) = (a, a)$ and $(q', p') = (a, -a+f(a))$,
Fig.~\ref{fig:partial_action}.c.
Finally, for twist maps, Eq.~(\ref{math:twist}), the Danilov theorem
Eq.~(\ref{math:Danilov}) gives $\nu$, as expected.

Let us now consider several non-trivial examples.  Linear maps are presented in Appendix A.

\subsection{McMillan map}

As our first example, we will consider the so-called McMillan map 
\cite{mcmillan1971problem},
\begin{equation}
\label{math:McMap}
\begin{bmatrix}
q'	\\ p'
\end{bmatrix}
=
\begin{bmatrix}
p	\\ -q + a\,p/(b\,p^2+1)
\end{bmatrix}.
\end{equation}
This map has been considered in detailed in Ref. \cite{Iatrou_2002}.
To illustrate the Danilov theorem, we will limit ourselves to a case with
$b > 0$ and $|a| < 2$, which corresponds to stable motion at small amplitudes.
Mapping~(\ref{math:McMap}) has the following integral:
\begin{equation}
\label{math:McM_Inv}
    \K(q,p) = b\,q^2 p^2 + q^2 + p^2 - a\,q\,p,
\end{equation}
which is non-negative for the chosen parameters.

We first notice that for small amplitudes, $b\,p^2 \ll 1$, 
this map can be approximated as
\begin{equation}
\label{math:McMap_approx}
\begin{bmatrix}
q'	\\ p'
\end{bmatrix}
\approx
\begin{bmatrix}
p	\\ -q + a\,p - a\,b\,p^3 + a\,b^2\,p^5 - ...
\end{bmatrix},
\end{equation}
and its zero-amplitude rotation number is \cite{courant1958theory}
\begin{equation}
	\nu(0) = \frac{1}{2\,\pi}\,\arccos\frac{a}{2}.
\end{equation}
At large amplitudes ($b\,p^2 \gg 1$), the rotation number becomes $0.25$.
We will now evaluate the rotation number analytically, using Eq.~(\ref{math:Danilov}):
Let us define a parameter,
\begin{equation}
w(\K) = 
\frac{1}{\sqrt{2}}\sqrt{1+\frac{d(\K)}{\sqrt{d(\K)^2+4\,\K\,b}}},
\end{equation}
which spans from 0 to 1 and where $d(\K)=a^2/4 + \K\,b - 1$.
Then, the rotation number can be expressed through Jacobi elliptic functions as 
follows:
\begin{equation}
\label{math:McMRotNum}
\nu(\K) = \frac{1}{4\,\mathrm{K}(w)}
	\ads \left( {\left( d(\K)^2+4\,\K\,b \right ) ^{-1/4}} , w \right),
\end{equation}
where $\mathrm{K}(w)$ is the complete elliptic integral of the first kind and 
the inverse Jacobi function, $\ads(x,w)$, is defined as follows
\begin{equation}
\ads(x,w) = \int_x^\infty \frac{\dd t}{\sqrt{(t^2+w^2)(t^2+w^2-1)}}.
\end{equation}
The rotation number, Eq.~(\ref{math:McMRotNum}), has the following series expansion:
\begin{equation}
\label{math:McMRotNumExpansion}
\nu(\K) \approx \nu(0) + \frac{3}{2\,\pi}\,\frac{b\,a}{\sqrt{(4-a^2)^3}}\K.
\end{equation}
\begin{figure}[t!]\centering
\includegraphics[width=\linewidth]{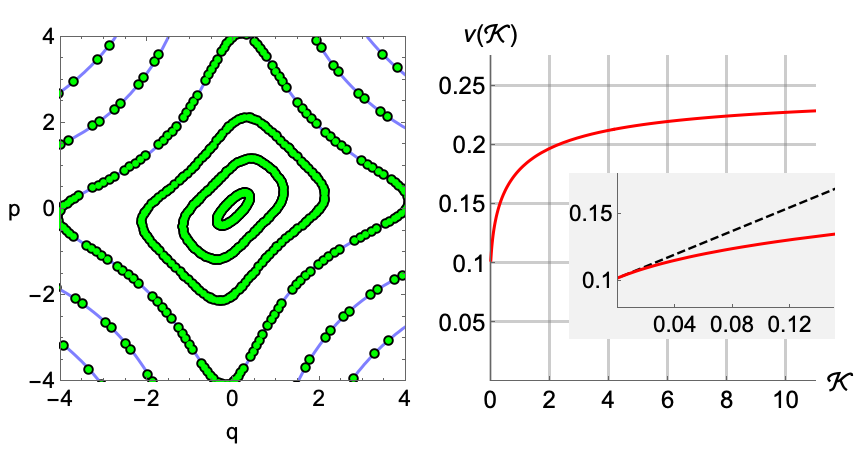}
\caption{\label{fig:McM}
	The left plot contains iterations (green dots) of the McMillan map
	($a = 1.6$, $b=1$).
	Constant level sets of the invariant are shown with blue lines.
	The right plot is the rotation number, Eq. (\ref{math:McMRotNum}), as a function of its integral, $\K$. 
	The inset shows the linear approximation, Eq. (\ref{math:McMRotNumExpansion}). 
	} 
\end{figure}
Figure~\ref{fig:McM} shows an example of the rotation number, for the case of
$a=1.6$ and $b=1$ ($\nu(0)\approx 0.102$), as a function of
integral, $\K$.

The McMillan invariant (\ref{math:McM_Inv}) also allows for an analytic evaluation of the 
action integral (\ref{math:J}).  We will omit the lengthy expressions, but will only present
a small-amplitude series expansion:
\begin{equation}
\label{math:action_mcm}
    J(\K) \approx \frac{\K}{\sqrt{4-a^2}} -
        \frac{b\,(2+a^2)\,\K^2}{\sqrt{(4-a^2)^5}}.
\end{equation}
Finally, we can also present a small-amplitude series expansion of the rotation number (\ref{math:McMRotNum}):
\begin{equation}
\label{math:McMRotNumApp}
\nu(J) \approx \nu(0) + \frac{3}{2\,\pi}\,\frac{b\,a}{4-a^2}J.
\end{equation}

\subsection{Cubic map}

\begin{figure}[t!]\centering
\includegraphics[width=\linewidth]{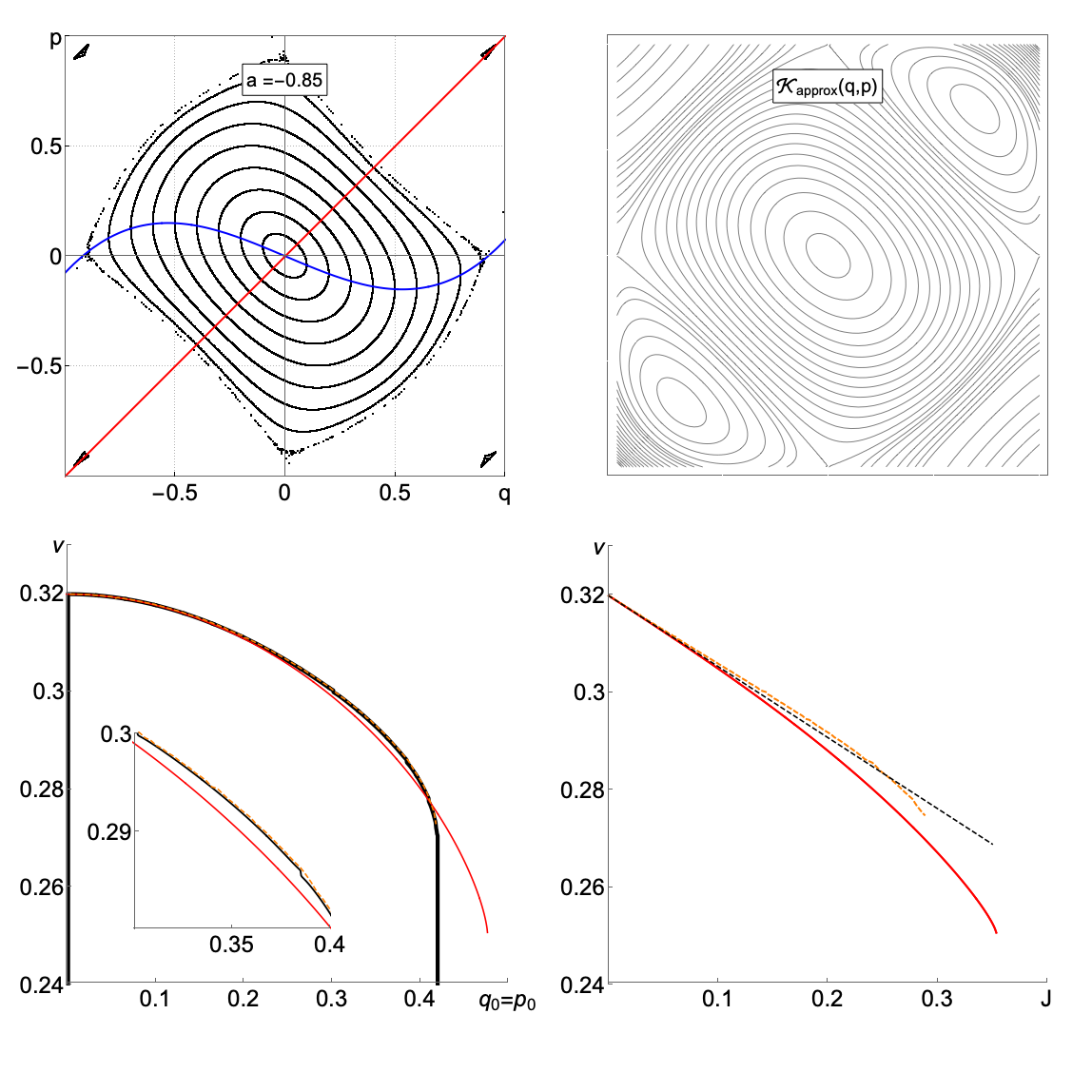}
\caption{\label{fig:CubicMap}
	The top row: phase-space trajectories of a cubic map, obtained by tracking with $a=-0.85$ (left plot) and level sets of the approximate invariant (\ref{math:OctApproxK}) (right plot), on the same scale.
	The red and blue lines in the top left plot corresponds to symmetry lines $p = q$ and $p = (a\,q+\epsilon\,q^3)/2$ respectively.
	The bottom row: the left plot shows the rotation number as a function of initial conditions in the form $q_0 = p_0$, by using the Eq.~(\ref{math:nu}) (black solid line), and by using the Danilov theorem, Eq.~(\ref{math:Dan2}) numerically (orange dashed). The red solid line corresponds to the rotation number obtained from the
	approximate invariant (\ref{math:OctApproxK}) using the Danilov theorem as well.
	The right bottom plot shows the dependence of $\nu$ as a function of action $J$,
	from tracking (orange dashed) and from the approximate invariant (\ref{math:OctApproxK}) (red solid).
	}
\end{figure}

As our second example, we will consider a non-integrable H\'enon cubic map \cite{Hénon:218956, Dullin}:
\begin{equation}
\label{math:OctMap}
\begin{bmatrix}
q'	\\ p'
\end{bmatrix}
=
\begin{bmatrix}
p	\\ -q + a\,p + \epsilon\,p^3
\end{bmatrix}.
\end{equation}
This map is well-known in accelerator physics as a symplectic octupole map.
At small amplitudes this map is linear and the rotation number is
\begin{equation}
	\nu \approx \frac{1}{2\,\pi}\,\arccos \left(\frac{a}{2}\right).
\end{equation}
At large amplitudes this map becomes chaotic and unstable.
Let us propose an approximate integral (the exact integral does not exist since
it is a non-integrable map).
\begin{equation}
\label{math:OctApproxK}
\begin{array}{l}
\displaystyle
\K_\text{c}(q,p) = p^2 + q^2 - a\,p\,q - \frac{\epsilon}{a}\,p^2q^2  \\[0.3cm]
\displaystyle \qquad +
    \frac{7\,\epsilon}{5\,a\,(4 - a^2)}\,
    \left( p^2 + q^2 - a\,p\,q \right)^2 +
    O\left( \epsilon^2 \right).
\end{array}
\end{equation}
The derivation of this approximate integral goes beyond the scope of this
article and will be described in subsequent publications.
For this illustration, the reader can verify by inspection that this integral is
approximately conserved, near the origin.
We will now use the Danilov theorem to evaluate the rotation number of this map
for various initial conditions with $q_0=p_0$.
Figure \ref{fig:CubicMap} shows the exact (numeric), Eq.~(\ref{math:nu}), and the approximate rotation number, calculated from (\ref{math:OctApproxK}) and (\ref{math:OctMap}) using the Danilov theorem, Eq.~(\ref{math:Danilov}).

A small-amplitude series expansion of the rotation number is:
\begin{equation}
\label{math:HenRotNumApp}
\nu(J) \approx \nu(0) - \frac{3}{2 \pi} \frac{\epsilon}{4-a^2}J,
\end{equation}
which is the same as in \cite{Dullin} and similar to Eq. (\ref{math:McMRotNumApp}).

\subsection{4-D integrable map}

In this section we will sketch out an example of how to use the Danilov theorem
to analyze an integrable multi-dimensional map.
Consider the following map, which can be realized in accelerators by employing
the so-called electron lens~\cite{ElectLens1, ElectLens2, Lobach:IPAC2018-THPAF071},
\begin{equation}
\label{math:4dMap}
\begin{bmatrix}
x'	\\[0.05cm] p_x' \\[0.05cm] y' \\[0.05cm] p_y'
\end{bmatrix}
=
\begin{bmatrix}
\alpha_x x + \beta\, p_x	                            \\
-\gamma_x x - \alpha_x\,p_x + \frac{a\,x'}{b\,r'^2+1}    \\
\alpha_y y + \beta\, p_y	                            \\
-\gamma_y y - \alpha_y\,p_y + \frac{a\,y'}{b\,r'^2+1}
\end{bmatrix},
\end{equation}
where $r^2 = x^2 + y^2$, $\beta\,\gamma_x=1+\alpha_x^2$,
$\beta\,\gamma_y=1+\alpha_y^2$, with $\alpha_x$, $\alpha_y$, $a$, $b$ and $\beta$ being some arbitrary parameters.
This map has two integrals of motion in involution (having a vanishing Poisson
bracket):
\begin{equation}
\label{math:Ang_rot_Inv}
    L = (\alpha_y-\alpha_x)\,x\,y + \beta(x\,p_y - y\,p_x)
\end{equation}
and
\begin{equation}
\label{math:McM_rot_Inv}
    \K = \Big(b + \frac{1}{r^2}\Big)\,T^2+\beta\,a\,T+r^2 + \frac{L^2}{r^2},
\end{equation}
where $T = \alpha_xx^2 + \alpha_yy^2 +\beta\,r\,p_r$ and $p_r=(x\,p_x+y\,p_y)/r$.
In order to employ the Danilov theorem, we must rewrite the map (\ref{math:4dMap})
in new variables, where this map is separated into two maps.
Such variables exist by virtue of this map being integrable.
We first notice that by introducing new variables,
\begin{equation}
    \label{math:sep_variables}
\begin{matrix}
\Tilde{x} = x/\sqrt{\beta}	
\\ \Tilde{p}_{x} = x\,\alpha_x/\sqrt{\beta} + p_x\sqrt{\beta}  
\\ \Tilde{y} = y/\sqrt{\beta}
\\ \Tilde{p}_{y} = y\,\alpha_y/\sqrt{\beta} + p_y\sqrt{\beta},
\end{matrix}
\end{equation}
the map (\ref{math:4dMap}) becomes symmetric in $\Tilde{x}$ and $\Tilde{y}$ with $\Tilde{a}=a\sqrt{\beta}$ and $\Tilde{b}=b\beta$.
The resulting map is separable in polar coordinates, $r$ and $\theta$, such that $x=r\,\cos(\theta)$ and
$y=r\,\sin(\theta)$, where we omitted the tilde ($\,\Tilde{}\,$) sign for clarity.  The resulting map is
\begin{equation}
\label{math:McMap_radial}
\begin{bmatrix}
r'	\\[0.3cm] p_r' \\[0.3cm] \theta' \\[0.3cm] p_\theta'
\end{bmatrix}
=
\begin{bmatrix}
\sqrt{p_r^2+\frac{p_{\theta}^2}{r^2}}	        \\[0.25cm]
-p_r \frac{r}{r'} + \frac{a\,r'}{b\,r'^2+1}     \\[0.25cm]
\theta + \arctan \frac{p_\theta}{r\,p_r}        \\[0.25cm]
p_\theta
\end{bmatrix},
\end{equation}
where the angular momentum $p_{\theta} = x\,p_y - y\,p_x = \const$ is the integral
of the motion.
An additional integral is
\begin{equation}
\label{math:McM_rad_Inv}
    \K(r,p_r, p_{\theta}) = b\,r^2 p_r^2 + r^2 + p_r^2 - a\,r\,p_r + \frac{p_{\theta}^2}{r^2}.
\end{equation}
Now we will use the Danilov theorem to obtain two unknown rotation numbers, $\nu_{\theta}$ and $\nu_r$. 
We first notice that $\K$ does not depend on $\theta$ and thus can be used to evaluate $\nu_r$ in Eq.~(\ref{math:Dan6})
directly, by treating $p_{\theta}$ as a parameter.
\begin{equation}
\label{math:freq_radial}
\begin{array}{l}
\ds\nu_r(\K, p_{\theta}) = \frac{\tau}{T_r} = 
    \frac{\int_{r}^{r'}
        \left(\frac{\pd \K}{\pd p_r}\right)^{-1}
    \,\dd r}
    {\oint
        \left(\frac{\pd \K}{\pd p_r}\right)^{-1}
    \,\dd r}                                \\[0.8cm]
\ds\qquad =
\mathrm{F}\left[
        \arcsin \sqrt{\frac{\zeta_3-\zeta_1}{\zeta_3+1}},
        \kappa
    \right]\Big/\left(2\,\mathrm{K}\left(
        \kappa
    \right)\right),
\end{array}
\end{equation}
where $\mathrm{K}(\kappa)$ is the complete elliptic integral of the first kind, 
$\mathrm{F}(\phi,\kappa)$ is the incomplete elliptic integral of
the first kind, elliptic modulus $\kappa$ is given by
\[
\kappa = \sqrt{\frac{\zeta_3-\zeta_2}{\zeta_3-\zeta_1}},
\]
and $\zeta_1 < 0 < \zeta_2 < \zeta_3$ are the roots of the polynomial
\[
\mathcal{P}_3(\zeta) =
    - \zeta^3 +
    \left[ \K + \left(\frac{a}{2}\right)^2 - 1 \right]\,\zeta^2+
    (\K - p_\theta^2)\,\zeta -
    p_\theta^2.
\]
In order to evaluate the angular rotation number, $\nu_\theta$, we first notice that there is some uncertainty 
as to which integral of the motion to employ: one can add an arbitrary function of $p_\theta$ to $\K$, $\K'=\K+f \left(p_\theta \right )$, to obtain another integral.  
This new integral of motion, $\K'$, gives the same $\nu_r$, but modifies the angular motion by some unknown linear function of time:
\begin{equation}
    \label{angular_motion}
    \frac{d \theta}{d t} = \frac{d \K'}{d p_\theta} = \frac{d \K}{d p_\theta} + f'\left(p_\theta \right ),
\end{equation}
\begin{equation}
    \label{angular_motion1}
    \theta(t) = \int \frac{d \K}{d p_\theta} dt + f'\left(p_\theta \right )t.
\end{equation}
Fortunately, we can resolve this uncertainty by using the angular portion of the map, Eq.~(\ref{math:McMap_radial}).
By its definition, the angular rotation number is
\begin{equation}
    \label{angular_freq1}
    \nu_\theta = \nu_r \frac{\Delta \theta \left ( T_r \right )}{2\,\pi},
\end{equation}
where 
\begin{equation}
    \label{angular_freq2}
    \Delta \theta \left ( T_r \right ) = \oint \frac{d \K}{d p_\theta} \left ( \frac{\partial \K}{\partial {p_r}} \right )^{-1} dr + k\,T_r,
\end{equation}
$k$ is an unknown coefficient and $T_r$ is the period of the radial motion,
\begin{equation}
\label{radial_period}
    T_r = \oint
        \left(\frac{\pd \K}{\pd p_r}\right)^{-1}
    \,\dd r.
\end{equation}
To determine the coefficient $k$ we will notice from
Eq.~(\ref{math:McMap_radial}) that $\Delta \theta ( \tau ) = \arctan \left ( \frac{p_\theta}{r\,p_r} \right )$.  Thus,
\begin{equation}
\label{coeff_k}
    k = \frac{1}{\tau} \left ( \arctan \left ( \frac{p_\theta}{r\,p_r} \right ) - \int_r^{r'} \frac{d \K}{d p_\theta} \left ( \frac{\partial \K}{\partial {p_r}} \right )^{-1} dr \right ) 
\end{equation} 
with 
\begin{equation}
\label{radial_time_step}
    \tau = \int_r^{r'}
        \left(\frac{\pd \K}{\pd p_r}\right)^{-1}
    \,\dd r.
\end{equation}
Now, recalling that $\nu_r = \tau / {T_r}$, we finally obtain
\begin{equation}
    \label{angular_freq3}
    \begin{array}{l}
         \ds  \nu_\theta = \frac{\nu_r }{2\,\pi} \oint \frac{d \K}{d p_\theta} \left ( \frac{\partial \K}{\partial {p_r}} \right )^{-1} dr \, +  \\ [0.35cm]
         \ds \frac{1}{2 \pi} \left ( \arctan \left ( \frac{p_\theta}{r\,p_r} \right ) - \int_r^{r'} \frac{d \K}{d p_\theta} \left ( \frac{\partial \K}{\partial {p_r}} \right )^{-1} dr \right ).
    \end{array}
\end{equation}
After some math, this expression can be rewritten as
\begin{equation}
\label{math:freq_angular}
\begin{array}{l}
\ds\nu_\theta(\K, p_{\theta}) =
    \frac{\Delta}{2\,\pi}\,\left[
    \nu_r - \frac{\Delta'}{\Delta} +
    \frac{\arctan\left(\frac{2\,p_\theta}{a}\,\frac{\zeta_3+1}{\zeta_3}\right)}{\Delta},
    \right]
\end{array}
\end{equation}
where
\[
\begin{array}{l}
\ds \Delta =
\frac{2\,p_\theta}{\zeta_3\,\sqrt{\zeta_3 - \zeta_1}}\,
    \Pi\left[
        \kappa\,
        \bigg|
        \frac{\zeta_3-\zeta_2}{\zeta_3}
    \right],                               \\[0.35cm]
\ds \Delta' =
    \frac{p_\theta}{\zeta_3\,\sqrt{\zeta_3 - \zeta_1}}\,
    \Pi\left[
        \arcsin \sqrt{\frac{\zeta_3-\zeta_1}{\zeta_3+1}},
        \kappa\,
        \Bigg|
        \frac{\zeta_3-\zeta_2}{\zeta_3}
    \right],
\end{array}
\]
and, $\Pi(\kappa\,|\alpha)$ and $\Pi(\phi,\kappa\,|\alpha)$ are the 
complete and the incomplete elliptic integrals of the third kind,
respectively.  One can note that for a linear 4D map ($b = 0$), we have $\nu_r = 2\,\nu_\theta$ for any value of $p_\theta$.
Fig.~\ref{fig:McM4D1} shows an example of the radial and the angular rotation numbers as a function of $\K$ for various values of $p_\theta$.

\begin{figure}[b!]\centering
\includegraphics[width=0.925\linewidth]{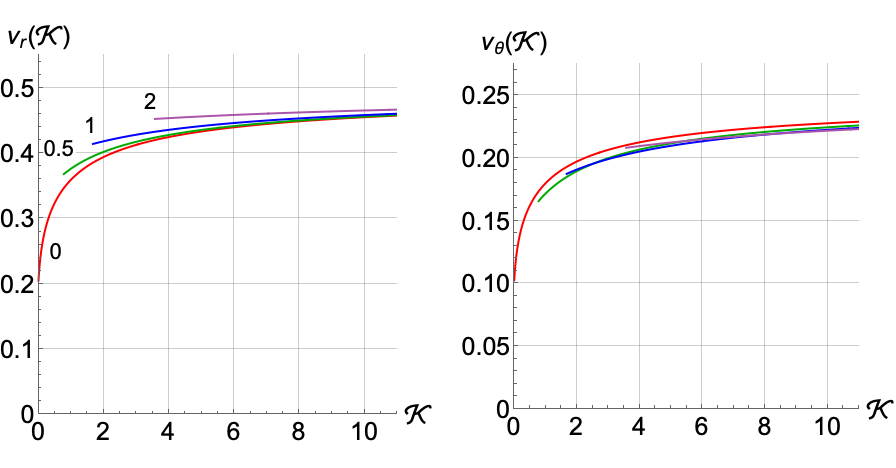}
\caption{\label{fig:McM4D1}
	Radial (left) and angular (right) rotation numbers as a function
	of the first integral of the map, $\K$, for different values of its second
	integral, $p_\theta$ (shown with color labels). The map	parameters are $b = 1$ and $a = 1.6$.
	Note that for $p_\theta = 0$,  $\nu_r = 2 \, \nu_\theta$, as expected, and equals to the frequency $\nu$ from the one-dimensional example of Fig.~\ref{fig:McM}
	} 
\end{figure}

\section{Summary}

In this paper we demonstrated a general and exact method of how to find a
Poincar\'e rotation number for integrable symplectic maps of a plane and its
connection to accelerator physics.
It complements the discrete Arnold-Liouville theorem for
maps~\cite{veselov1991integrable,arnold1968ergodic} 
and permits the analysis of dynamics for integrable systems.
Eq.~(\ref{math:hamilt}) also permits to express the Hamilton function of a given
integrable map explicitly.
Several examples were presented in our paper.
These examples demonstrate that the Danilov theorem is a powerful tool.
The McMillan integrable map is a classic example of a nonlinear integrable
discrete-time system, which finds applications in many areas of physics,
including accelerators~\cite{Antipov_2017, Lobach:IPAC2018-THPAF071}.
It is a typical member of a wide class of area-preserving transformations
called a twist map~\cite{meiss1992symplectic}.
For non-integrable maps, which are also very common in accelerator science,
this new theorem could allow for an approximate evaluation of rotation numbers,
provided there exists an approximate integral of motion, like
Eq.~(\ref{math:OctApproxK}).

\section{Acknowledgments}

The authors would like to thank Jeffrey Holmes and Stanislav Baturin for carefully 
reading this manuscript and for their helpful comments. This research is supported 
by Fermi Research Alliance, LLC under Contract No. 
DE-AC02-07CH11359 with the U.S. Department of Energy and by the University of Chicago.

\appendix
\section{Linear maps}
In this appendix we will consider two examples of linear 
maps and we will use Eq.~(\ref{math:Danilov}) for one and Eq.~(\ref{math:Dan7}) for the second one.

\subsection{Linear accelerator map}

Consider a linear symplectic map,
\begin{equation}
\label{math:LMap}
    \begin{bmatrix}
        q'	\\ p'
    \end{bmatrix}
    =
    \begin{bmatrix}
     a & b		\\
     c & d
    \end{bmatrix}
    \begin{bmatrix}
        q	    \\ p
    \end{bmatrix},
\end{equation}
with $a\,d-b\,c=1$ and $|a+d| \leq 2$.
This map is very common in accelerator physics and has been described in 
\cite{courant1958theory}.
The rotation number (the betatron frequency) for this map is well known:
\begin{equation}
\label{math:Nu_linear}
	\nu = \frac{1}{2\,\pi}\,\arccos\frac{a+d}{2}.
\end{equation}
To obtain this equation using the Danilov theorem, we will recall that this map
has the following Courant-Snyder integral (invariant): 
\begin{equation}
	\K = c\,q^2 + (d-a)\,q\,p - b\,p^2.
\end{equation}
Let us assume that $c > 0$, then $b \leq 0$ and $\K(q,p) \geq 0$ for any $q$ and
$p$.
From this, we obtain 
\begin{equation}
    \left( \frac{\pd\K}{\pd p} \right)^{-1} = \frac{\pd p}{\pd \K} =
    \frac{\pm 1}{\sqrt{\left[ (a+d)^2-4\right]\,q^2-4\,b\,\K}}.
\end{equation}
We will use 
\begin{equation}
(q,p) = (\sqrt{\K/c}, 0)
\end{equation}
and
\begin{equation}
(q',p') = (b\,\sqrt{\K/c},\sqrt{\K\,c})
\end{equation}
After a straightforward evaluation of integrals in Eq.~(\ref{math:Danilov}), we
obtain:
\begin{equation}
	\nu = \frac{1}{2\,\pi}\,\arccos \frac{a+d}{2},
\end{equation}
same is in Eq.~(\ref{math:Nu_linear}).

\subsection{\label{sec:Ex2}Brown map}

As a second example we will consider 
the Brown map \cite{brown1983map,brown1993map}, $\M_\mathrm{B}$,
\begin{equation}
\label{math:Brown}
    \begin{bmatrix}
        q'	\\ p'
    \end{bmatrix}
    =
    \begin{bmatrix}
        p     \\ -q + |p|
    \end{bmatrix},
\end{equation}
which has the following integral,
\begin{eqnarray}
\label{math:Brown1}
\nonumber
\K(q,p) & = & \frac{1}{8} \Bigg(
q +
p +
\big| q - |p| \big| +
\big| p - |q| \big| +					\\\nonumber
&&
2\,\Big| q - \big| p - |q| \big| \Big| + 	
2\,\Big| p - \big| q - |p| \big| \Big| + 		\\\nonumber
&&
\bigg| q - |p| + \Big| p - \big| q - |p| \big| \Big| \bigg| +	\\
&&
\bigg| p - |q| + \Big| q - \big| p - |q| \big| \Big| \bigg| \Bigg).
\end{eqnarray}

The map has only one stable fixed point, located at the origin, with $\K=0$.
Constant level sets of $\K>0$ are polygons, geometrically similar to each other, with 9 sides, labeled by Roman
numerals, see Fig.~\ref{fig:Knuth2}.a.
All orbits belonging to these levels are periodic with
\begin{equation}
    	\M_\mathrm{B}^9(q,p) = (q,p),
\end{equation}
and in fact, they are permutation 9-cycles such that
\begin{eqnarray*}
\ldots\rightarrow
\text{I}		&\rightarrow&
\text{III}	\rightarrow
\text{V}	\rightarrow
\text{VII}	\rightarrow
\text{IX}	\rightarrow			\\
			&\rightarrow&
\text{II}	\rightarrow
\text{IV}	\rightarrow
\text{VI}	\rightarrow
\text{VIII}	\rightarrow
\text{I}	\rightarrow			\ldots.
\end{eqnarray*}

\begin{figure}[th!]\centering
\includegraphics[width=\linewidth]{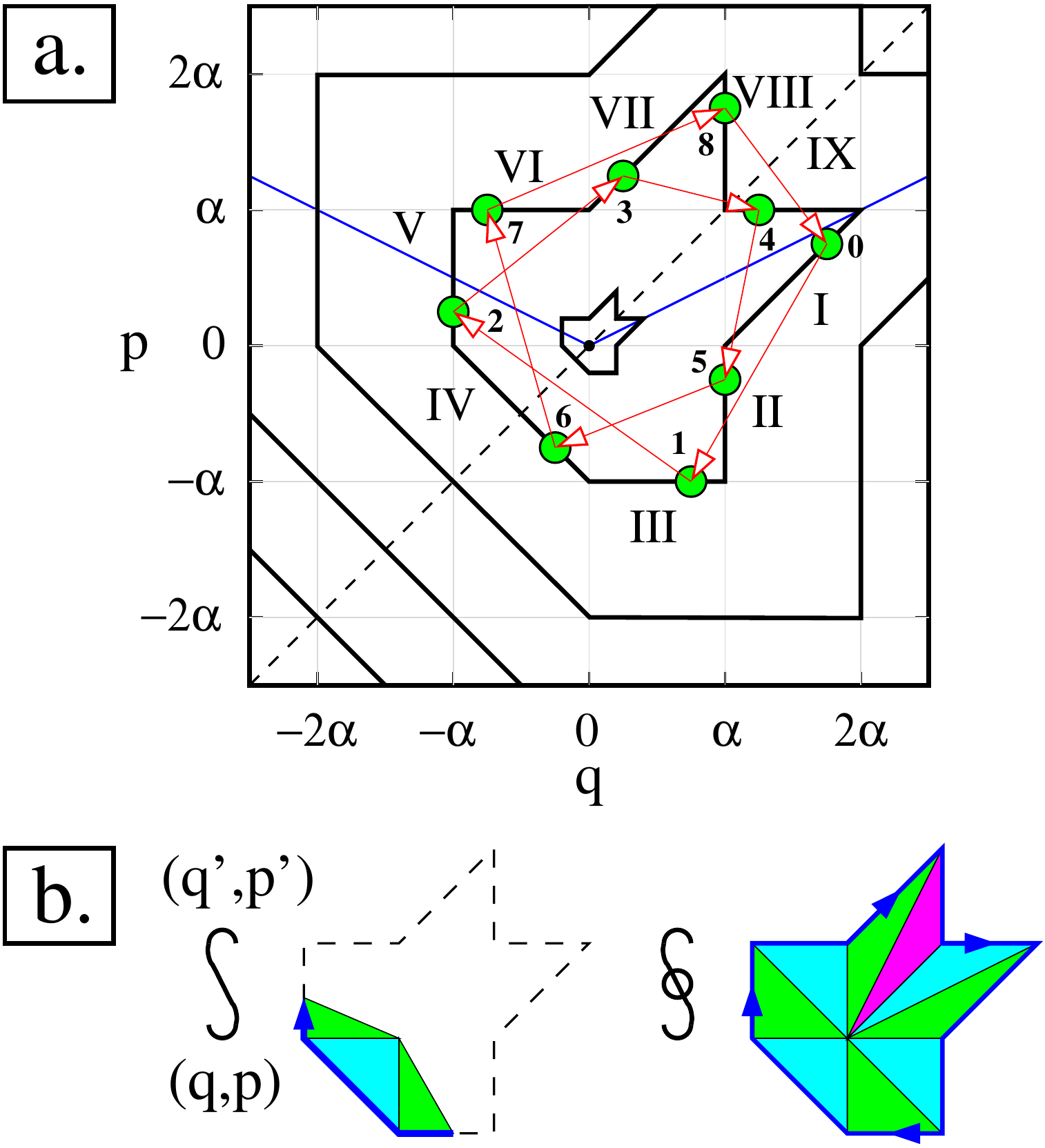}
\caption{Brown map.
	(a.) Constant level sets of the invariant, $\K(q,p) = \const$
	(black solid polygons).
	Dashed black line $p=q$ and blue line $p = \frac{1}{2}\,|q|$ illustrate
	two reflection symmetries of the invariant polygons.
	Line segments are labeled with Roman numerals.
	Green points are an example of a 9-cycle orbit, where the Arabic numerals show the iteration number.
	(b.) An example of possible contour of integration for the numerator and denominator in Danilov theorem.
	}
\label{fig:Knuth2}
\end{figure}
Since it is a linear map ($\nu = \const$ for all orbits), we will use Eq. (\ref{math:Dan7}) to determine its rotation number.  
It is obvious from Fig. \ref{fig:Knuth2}.b that $J = 4.5\,\alpha$, while $J' = 1 \,\alpha$, 
where $\alpha$ is some arbitrary scale parameter, resulting in $\nu = \frac{2}{9}$.

\section{Numerical procedure for Danilov theorem}

In this appendix we will consider two numerical procedures, which can be employed
in order to use Eq.~(\ref{math:Dan2}) for mappings in McMillan form when only the 
mapping equation is known or when we have an approximate (or an exact) invariant of 
the motion but we can not compute action integrals analytically.

We will start with the case when we have only the mapping equations.
As a first step we will rewrite the map in polar coordinates
\[
    q = r\,\cos\phi,    \qquad\qquad    p = r\,\sin\phi.
\]
Then we will iterate for various initial conditions $q_\text{ini}^{(k)}$,
let say in a form of $q_\text{ini}^{(k)} = q_0^{(k)} = p_0^{(k)}$, so that we have
a collection of points in a form
\[
     (r_0^{(k)},\phi_0^{(k)}),(r_1^{(k)},\phi_1^{(k)}),(r_2^{(k)},\phi_2^{(k)}),
     \ldots,(r_n^{(k)},\phi_n^{(k)}).
\]
We can then sort each orbit such that
\[
    \widetilde{\phi}_0^{(k)} < \widetilde{\phi}_1^{(k)} <
    \widetilde{\phi}_2^{(k)} < \ldots <
    \widetilde{\phi}_n^{(k)},
    \vspace{0.1cm}
\]
where $(\widetilde{r}_i^{(k)},\widetilde{\phi}_i^{(k)})$ are the points of a new sorted $k$-th orbit.
Now, for each orbit we can compute the action and the partial action numerically as
\begin{equation}
    J^{(k)} = \frac{1}{2\,\pi}\,\sum_{i=0}^n
        \frac{\left(\widetilde{r}_i^{(k)}\right)^2}{2}\,
        \left[\widetilde{\phi}_i^{(k)} - \widetilde{\phi}_{i-1}^{(k)}\right]
\end{equation}
and
\begin{equation}
    J'^{(k)} = \frac{1}{2\,\pi}\,\sum_{\pi/2 < \widetilde{\phi}_i^{(k)} < \pi}
        \frac{\left(\widetilde{r}_i^{(k)}\right)^2}{2}\,
        \left[\widetilde{\phi}_i^{(k)} - \widetilde{\phi}_{i-1}^{(k)}\right]
\end{equation}
respectively.
Finally, using the Danilov theorem, we can find the rotation number as a numerical derivative
\begin{equation}
    \nu^{(k)} = \frac{J'^{(k+1)} - J'^{(k)}}{J^{(k+1)} - J^{(k)}}.
\end{equation}

If one would like to apply the Danilov theorem directly to an approximate or exact invariant
of motion, we can proceed in a similar manner.
First, we rewrite the invariant of motion in polar coordinates,
$\K_\text{approx}(r,\phi)$.
Then, for different values $\K_\text{approx}^{(k)}$ we will numerically solve $n$ equations
\begin{equation}
     \K_\text{approx}(r,\phi^{(k)}_i) = \K_\text{approx}^{(k)}
\end{equation}
with $\phi^{(k)}_i = 2\,\pi\,i/n$ and $i=0,1,\ldots,n-1$.
Denoting the smallest positive root of equation above as
$r^{(k)}_i$, we can find action and partial actions as
\begin{equation}
    J^{(k)} = \frac{1}{2\,\pi}\,\sum_{i=0}^{n-1}
        \frac{\left(r_i^{(k)}\right)^2}{2}\,
        \left[\phi_i^{(k)} - \phi_{i-1}^{(k)}\right]
\end{equation}
and
\begin{equation}
    J'^{(k)} = \frac{1}{2\,\pi}\,\sum_{\pi/2 < \phi_i^{(k)} < \pi}
        \frac{\left(r_i^{(k)}\right)^2}{2}\,
        \left[\phi_i^{(k)} - \phi_{i-1}^{(k)}\right],
\end{equation}
along with the rotation number 
\begin{equation}
    \nu^{(k)} = \frac{J'^{(k+1)} - J'^{(k)}}{J^{(k+1)} - J^{(k)}}.
\end{equation}

\bibliographystyle{apsrev4-1}
\bibliography{bibfile}

\end{document}